\documentclass{article}
\usepackage[utf8]{inputenc}
\usepackage[T1]{fontenc}
\usepackage{amsmath,amssymb,amsthm}
\usepackage{mathtools}
\usepackage{booktabs}
\usepackage{float}
\usepackage{hyperref}
\usepackage[margin=2.5cm]{geometry}
\usepackage{natbib}% Citation support using natbib.sty
\bibpunct[, ]{(}{)}{;}{a}{,}{,}% Citation support using natbib.sty

\usepackage{lineno}
\usepackage{enumitem}

\usepackage{tcolorbox}
\usepackage{graphicx}
\usepackage{xcolor}

\newtheorem{theorem}{Theorem}

\newtheorem{remark}{Remark}

\begin{document}

\title{AI Detectors Fail Diverse Student Populations: A Mathematical Framing of Structural Detection Limits}

\author{N.A. Garland \\ \small Queensland Quantum and Advanced Technologies Research Institute, Griffith University, QLD, Australia \\ \small School of Environment and Science, Griffith University, QLD, Australia }
% \affil{home}
% \affiliation{Queensland Quantum and Advanced Technologies Research Institute, Griffith University, Nathan 4111, QLD, Australia}
% \affiliation{School of Environment and Science, Griffith University, Nathan 4111, QLD, Australia}

% \author{
% \name{Author\textsuperscript{$\dagger$}\thanks{$\dagger$ Email: anon@ymo.us}}
% }

\date{\today}

\maketitle
\begin{abstract}
Student experiences and empirical studies report that ``black box'' AI text detectors produce high false positive rates with disproportionate errors against certain student populations, yet typically theoretical analyses model detection as a test between two known distributions for human and AI prose. This framing omits the structural feature of university assessment whereby an assessor generally does not know the individual student's writing distribution, making the null hypothesis composite. Standard application of the variational characterisation of total variation distance to this composite null shows trade-off bounds that any text-only, one-shot detector with useful power must produce false accusations at a rate governed by the distributional overlap between student writing and AI output. This is a constraint arising from population diversity that is logically independent of AI model quality and cannot be overcome by better detector engineering or technology. A subgroup mixture bound connects these quantities to observable demographic groups, providing a theoretical basis for the disparate impact patterns documented empirically. We propose suggestions to improve policy and practice, and argue that detection scores should not serve as sole evidence in misconduct proceedings.
\end{abstract}

% \begin{keywords}
% academic integrity; AI detection; composite hypothesis testing; higher education; 
% \end{keywords}

% \newpage
%=============================================================================
\section{Introduction}
\label{sec:intro}
%=============================================================================

The deployment of AI text detection tools in university assessment has become widespread, driven by legitimate concerns about academic integrity in the era of large language models (LLMs) (see~\citet{Tight2024,LuoJess2024}). However, a growing body of empirical work has documented serious problems with these tools. For example, \citet{weberwulff2023} tested 14 detection tools and found that none achieved 80\% accuracy, with all tools producing false positives (human-written text flagged as AI-generated) and false negatives (AI-generated text classified as human-written). \citet{liang2023} demonstrated that detectors disproportionately flag non-native English speakers, raising equity concerns. \citet{elkhatat2023} found inconsistencies across detection tools when applied to human-written control samples. Most recently, \citet{hadra2026} evaluated detector accuracy for institutions serving English as a Foreign Language learners, concluding that current tools are ``unsuitable as authoritative arbiters of authorship'' and that institutions should approach detection ``with caution and nuance.'' The broader implications of LLMs for assessment validity and fairness have been recognised by the educational measurement community~\cite{Hao2024}, but a formal statistical framework connecting population diversity to structural detection limits has not yet been developed.

These empirical findings paint a consistent picture: current AI text detectors are unreliable, particularly for diverse student populations. But the existing literature largely treats each study's findings as properties of \emph{particular} detectors that might be improved by better engineering. In this paper, we ask a different question: are there structural reasons why any \emph{text-only, one-shot} detector must face a trade-off between detection power and false accusations, independent of detector quality?

We show that the answer is yes, and that the source of this structural barrier is \emph{population diversity}---the natural variation in writing style, ability, and linguistic background across a student body. This is distinct from, and complementary to, the widely-discussed problem of AI text quality converging toward human text \citep{sadasivan2023,chakraborty2024}.

Existing theoretical analyses \citep{sadasivan2023,chakraborty2024} model AI detection as a binary test between two \emph{known} distributions: $p_H$ (``the'' human distribution) versus $p_M$ (``the'' AI distribution). They show that as these distributions converge, detection approaches random guessing.

But in a university, there is no single ``human distribution.'' Each student writes differently, and the assessor typically does not know an individual student's writing characteristics when evaluating a submission. The null hypothesis is not ``this was drawn from $p_H$'' but rather:

\begin{center}
\emph{``This document was written by student $i$, whose writing distribution $p_{\theta_i}$ is unknown.''}

\end{center}

In the language of statistical hypothesis testing \citep{lehmann2005}, this is a \emph{composite} null hypothesis - the null distribution is not a single known entity but a family $\{p_\theta : \theta \in \Theta\}$ indexed by student characteristics. This seemingly simple reframing has concrete, quantifiable consequences for what any detector can achieve.

In this work we present our arguments towards modelling a ``black-box'' university one-shot AI detection as a composite hypothesis test with a composite null, making explicit a structural feature that prior theoretical work overlooks. We identify population diversity as a source of detection difficulty that is logically independent of AI quality, and show precisely what policy levers can address each mode (Section~\ref{sec:twosources}). We show a bound that helps to connect the theoretical overlap quantities to observable, real student subgroups, providing a bridge to the bias and auditing evidence documented by \citet{liang2023}, \citet{weberwulff2023}, \citet{elkhatat2023}, and \citet{hadra2026}. We conclude by proposing a concrete stratified FPR auditing procedure that institutions can implement with existing data (Section~\ref{sec:audit}).

The proof techniques are standard applications of the total variation (TV) distance variational inequality. We make no claim to have discovered new results in hypothesis testing theory. This present contribution is the applied formalisation, the conceptual separation of failure modes, and the connection to institutional practice enabled by relatively elementary mathematical exercises with fundamental statistics definitions.

%=============================================================================
\section{Background}
\label{sec:prior}
%=============================================================================

We briefly review the two main theoretical results in the literature, highlighting the structural assumption that each makes. The key mathematical tool is the \emph{total variation} (TV) distance between two probability distributions $P$ and $Q$, defined as 
\begin{align}
    \mathrm{TV}(P,Q) = \sup_A |P(A) - Q(A)| = \frac{1}{2}\int_{\mathcal{X}} |p(x) - q(x)|\,d\mu(x), 
\end{align}
where $\sup$ is the supremum, and $A$ is any event possible in the set of all possible text documents (the sample space), $\mathcal{X}$, with reference measure $\mu(x)$.

For any test function, which we will shortly formalise as our detector, $\phi: \mathcal{X} \to [0,1]$, the variational inequality gives
\begin{equation}
\label{eq:variational}
|\mathbb{E}_P[\phi] - \mathbb{E}_Q[\phi]| \;\leq\; \,\mathrm{TV}(P,Q),
\end{equation}
where $\mathbb{E}_P[\phi]=\int_{\mathcal{X}}\phi(x)p(x)d\mu(x)$ is the expectation of $\phi$ over a given distribution $P$. This single inequality is the only mathematical tool used in our analysis. 

In terms of distributional convergence, \citet{sadasivan2023} applies~\eqref{eq:variational} to a known human distribution $P_H$ and AI distribution $P_M$, and show any detector's power exceeds its false positive rate by at most $\,\mathrm{TV}(P_H, P_M)$. As AI models improve and $\mathrm{TV}(P_M, P_H) \to 0$, all detection metrics degrade. A key assumption was that both $P_H$ and $P_M$ are single, known distributions.

For multi-sample detection, the work of~\citep{chakraborty2024} demonstrates reliable detection requires $O(1/\delta^2)$ independent samples when $\mathrm{TV}(p_M, p_H) = \delta > 0$. In the one-shot setting ($n=1$), this provides no useful guarantee when $\delta$ is small. Again, their study assumed both distributions were known with multiple independent and identically distributed samples available.

The structural gap is that neither result addresses the case where the null distribution varies across the population and is unknown for each individual. This is precisely the situation in university assessment.

The one-shot AI detection problem is also structurally related to authorship verification without enrolment in computational stylometry \citep{koppel2004,stamatatos2009,Dawson2019}. In that setting, one must determine whether a questioned document was written by a claimed author, given no (or very limited) reference material from that author.

Both problems have a composite null in that the writer's distribution is unknown. The fundamental difficulty of verification without enrolment data is well-established in the stylometry community \citep{koppel2004,stamatatos2009}.

The key structural asymmetry is that the \emph{alternative} $p_M$ is, in principle, \emph{known} as a detector software designer can query AI systems and characterise their output distribution, and define rules. Classical authorship verification faces a \emph{doubly composite} problem (unknown claimed author \emph{and} unknown alternative author). Our Results~1-3 exploit this known-alternative structure. In particular, Result~3 uses the fact that $p_M$ can be sampled to construct the subgroup-vs-AI classifier and to generate the AI-side corpus needed for the auditing protocol. This asymmetry does not make detection \emph{easy}, the composite null still governs the FPR trade-off, but it enables a specific form of auditing we propose.

%=============================================================================
\section{Methods}
\label{sec:setup}
%=============================================================================

\subsection{Setting the University Detection Problem}

A university receives a single document $x$ from student $i$ for assessment task $\tau$. The assessor wishes to determine whether student~$i$ wrote $x$ or whether it was substantially generated by an AI system. Each student's writing is characterised by an unknown parameter $\theta_i \in \Theta$ encoding writing ability, domain knowledge, style, effort, and linguistic background. The student generates documents from distribution $p_{\theta_i,\tau}$, while a machine generated document generates from $p_{M,\tau}$. All bounds in this paper are per-task, they apply to a fixed task $\tau$ and its associated distributions. We suppress $\tau$ for notational simplicity except where task dependence is the focus (Section~\ref{sec:task}).

\subsection{ Composite Null}

The null hypothesis is
\[
H_0 : x \sim p_\theta \quad \text{for some unknown } \theta \in \Theta.
\]
A \emph{detector} is a measurable function $\phi: \mathcal{X} \to [0,1]$. For a document $x$, $\phi(x)$ is the probability with which the detector declares $x$ to be AI-generated. We define the per-student false positive rate (FPR)
$\alpha(\theta;\phi) = \mathbb{E}_{p_\theta}[\phi(x)]$, and the detector power (true positive rate / TPR) is defined as  $\beta(\phi) = \mathbb{E}_{p_M}[\phi(x)]$. Let $\pi$ be a measure of probability on $\Theta$, and for any measurable subset $S\subseteq\Theta$, $\pi(S)$ is the fraction of students whose type $\theta$ falls in $S$. 

We define the overlap set for some TV-distance threshold $\delta$, $\Theta^*(\delta) = \{\theta \in \Theta : \mathrm{TV}(p_\theta, p_M) \leq \delta\}$, as the set of students whose writing is within TV-distance $\delta$ of AI output. All overlap quantities are generally considered task-specific, such that $\Theta^*$, $\delta^*$, and subgroup mixture TV distances depend on $\tau$. Figure~\ref{fig:schematic} illustrates the structure of the composite null and how it differs from the simple-vs-simple hypothesis setting generally assumed in prior work.

\begin{figure}[h]
\centering
    \includegraphics[width=0.71\linewidth]{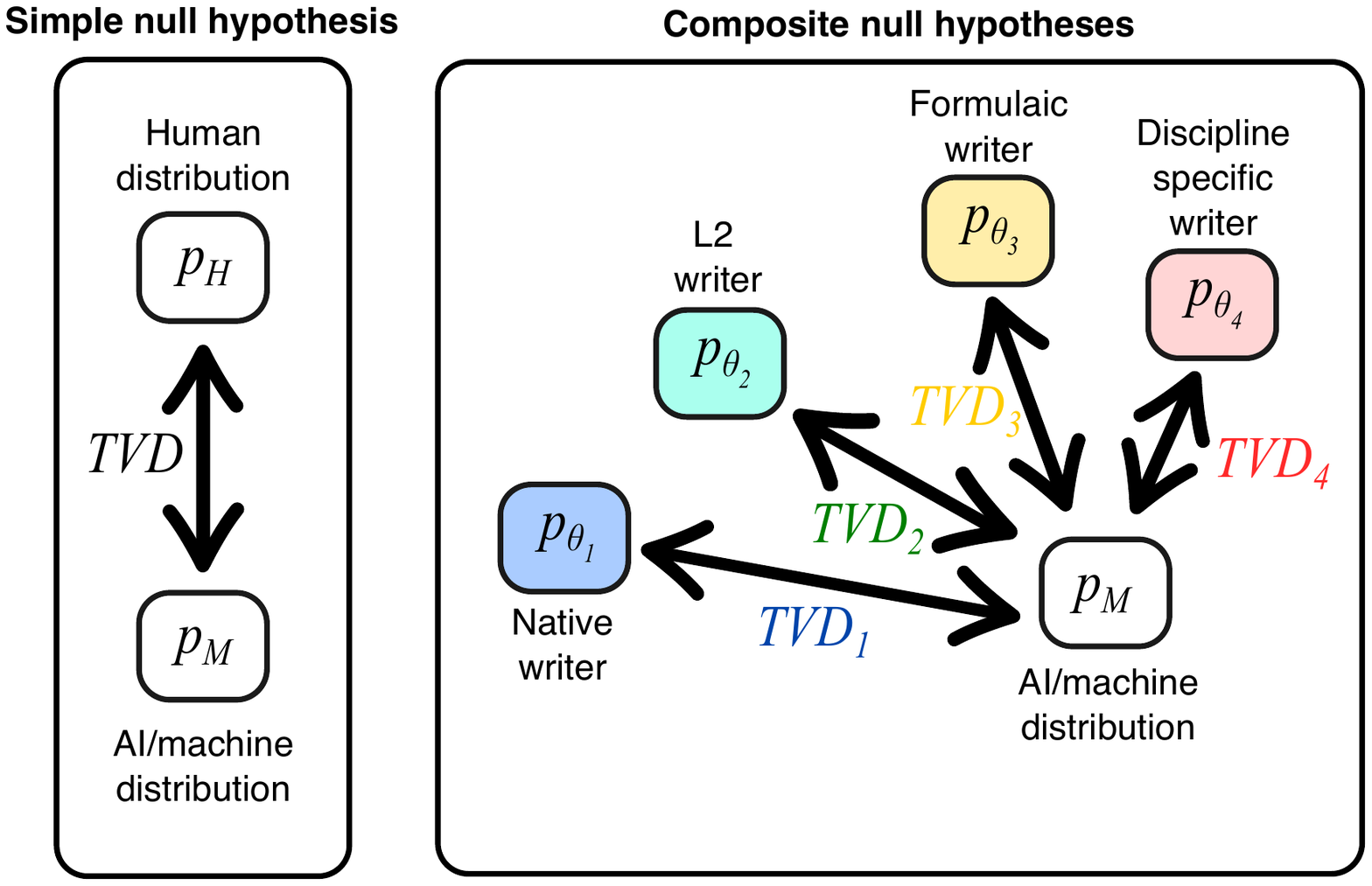}
\caption{Schematic comparison of the testing structure. \textbf{Left:} prior theoretical work assumes a single known human distribution $p_H$ and single known AI distribution $p_M$. \textbf{Right:} the university setting involves a \emph{composite null}---each student has a different writing distribution $p_\theta$, and the detector does not know which $\theta$ applies. Some students (e.g. those where English is a second language, or students writing a discipline specific assessment that uses structured norms) may have distributions very close to $p_M$ in total variation distance.}
\label{fig:schematic}
\end{figure}

%=============================================================================
\subsection{Theoretical Results Any Detector Must Face}
\label{sec:results}
%=============================================================================

We now summarise the three main mathematical results. Each is a consequence of the variational inequality~\eqref{eq:variational} applied to the composite null structure. Formal proofs are given in Appendix~\ref{app:proofs} for interested readers.

\subsubsection{An Average-Case Trade-Off}

\begin{tcolorbox}[colback=blue!5!white, colframe=blue!50!black, title=Result 1: Average-case size--power trade-off]
Suppose a fraction $\pi(\Theta^*)$ of students have writing distributions within TV-distance $\delta$ of the AI output. Then \textbf{any} detector with power (true positive rate) $\beta_0$ incurs a population-averaged false positive rate, $\bar{\alpha}$, of at least
\[
\bar{\alpha} \geq \pi(\Theta^*) \cdot (\beta_0 - \delta).
\]
\end{tcolorbox}

To demonstrate the meaning of this result, we present a conceptual example. Consider a case where 10\% of students write in a manner that is close to AI output, indicated by a nominal TV distance of 0.05 from expected AI distribution  ($\pi(\Theta^*) = 0.10$, $\delta = 0.05$), then a detector with 80\% power ($\beta_0=0.8$) must produce, on average, at least $\bar{\alpha}\geq 0.10 \times (0.80 - 0.05) \geq 0.075$ false accusations \emph{per student tested}. At an institution with 10,000 students, this is may be approximately 750 false accusations, and not because the detector is poorly designed, but because of the mathematical structure of the problem. 

We note these parameters are illustrative for this example in order to reproduce FPR on the order of 10\%-20\% demonstrated empirically in literature. The actual values of $\pi(\Theta^*)$ and $\delta$ are empirically unknown, and vary across institutions and disciplines, and cannot be estimated from the theory alone. We emphasise our example demonstrates the \emph{form} of the trade-off, not a quantitative prediction tool to use at a specific institution, and that this result is best understood as an interpretative lens through which we can explain empirical results and feedback from higher education communities reporting high FPR cases.

To help visualise the conceptual trade off, and sharp activation of non-negligible FPR values, Figure~\ref{fig:bound_map} displays the lower bound from Result~1 across the parameter space. The bound exceeds typical institutional tolerance thresholds of $1$--$5\%$ once $\pi(\Theta^*)$ surpasses approximately $5\%$ and $\delta$ is moderate. While neither $\pi(\Theta^*)$ nor $\delta$ has been directly measured in these distributional terms, the empirical literature constrains them to a regime where the bound is informative. For the overlap fraction $\pi(\Theta^*)$: international students constitute $20$--$33\%$ of enrolments at diverse institutions~\citep[e.g.,][]{HESA2025, AusDeptEd2024}, and \citet{liang2023} found that seven widely-used detectors misclassified $61.3\%$ of TOEFL essays by non-native English speakers as AI-generated, implying that a substantial share of L2 writers produce text that existing detectors cannot distinguish from AI output. Even conservatively, this places the effective overlap fraction at diverse universities in the range $5$--$20\%$. For the TV threshold $\delta$, \citet{sadasivan2023} estimated aggregate human-vs-AI total variation distances of roughly $0.2$--$0.4$ for $300$-token passages using GPT-3 models, with smaller values for more capable models and more constrained domains. Since the aggregate TV is an upper bound on the per-student TV for students within the overlap set, individual $\delta$ values below $0.1$ are plausible for L2 writers on constrained assessment tasks. The illustrative example ($\pi(\Theta^*) = 0.10$, $\delta = 0.05$) thus sits within an empirically plausible regime, though we emphasise these are consistency arguments rather than direct estimates. The key observation from the figure is structural: the bound becomes non-negligible across a broad and empirically realistic region of the parameter space, not merely at extreme or contrived parameter values.

 % This bound requires $\beta_0 > \delta$ to be nontrivial: a detector with no power ($\phi \equiv 0$) trivially has zero false positives. The result bounds the \emph{joint} achievable region, not either quantity alone. It is a trade-off bound, not an unconditional impossibility result.

 \begin{figure}[h]
\centering
    \includegraphics[width=0.71\linewidth]{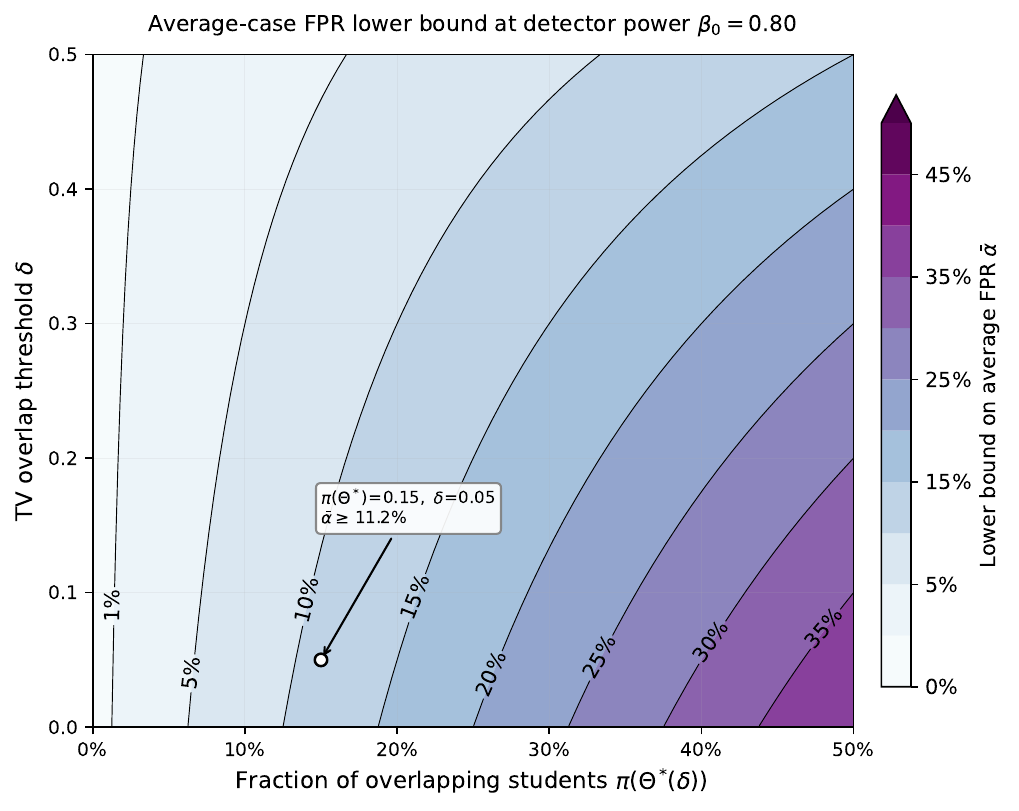}
\caption{The lower bound on population-averaged FPR from Result 1 at detector power $\beta_0 = 0.80$, shown across the $(\pi(\Theta^*(\delta)), \delta)$
 parameter space. Note that $\pi(\Theta^*(\delta))$
is non-decreasing in $\delta$ by construction, so a given student population traces a monotone curve through this space rather than occupying it freely. The figure shows the bound value at each mathematically valid parameter pair; the \textit{realised} pair for any specific population and task is an empirical question. The illustrative example from Section 3.3.1 is marked.}
\label{fig:bound_map}
\end{figure}

\subsubsection{Worst-Case Protection}

\begin{tcolorbox}[colback=blue!5!white, colframe=blue!50!black, title=Result 2: Worst-case size--power bound]
If the university requires that \textbf{no individual student} face a false positive rate exceeding $\alpha_0$, then the detector's power is bounded
\[
\beta(\phi) \;\leq\; \alpha_0 + \delta^*,
\]
where $\delta^* = \inf_\theta \mathrm{TV}(p_\theta, p_M)$ is the TV distance between the AI output and the \textbf{nearest} human writer.
\end{tcolorbox}

As we did before, let us present another conceptual example to demonstrate the meaning of the above result. If the nearest human writer in a student cohort is at TV distance $\delta^* = 0.05$ from the AI generated text model, and the university demands a false positive rate below 1\% ($\alpha_0 = 0.01$) for every student, then the detector's power cannot exceed $0.01 + 0.05 = 0.06$; it catches at most 6\% of AI-generated submissions. This is not a limitation of any particular detector, it is a mathematical consequence of the problem nature.

We wish to note that Result~2 is a necessary condition on any valid test, and not a minimax theorem in the sense of \citet{huber1965} and \citet{huber1973}. A minimax result would additionally construct least-favourable distributions and prove a saddle-point exists \citep[see][]{fauss2021}. Our bound gives an upper limit on achievable power under a size constraint but does not identify optimal tests.

\subsubsection{Subgroup Mixture Bound}

The overlap quantities $\pi(\Theta^*)$ and $\delta^*$ involve TV distances between individual-student distributions and the AI generated text model, these are not directly observable. Result~3 connects the interpretable, conceptual theory to \emph{observable} subgroup-level quantities.

Suppose students are partitioned into demographic or disciplinary subgroups $G_1, \ldots, G_K$ (e.g., by first language, degree programme, year of study). The \emph{subgroup mixture distribution} $\bar{p}_{G_k}$ is the distribution of documents produced by a randomly-chosen student from subgroup $G_k$, estimable from pooled submissions.

\begin{tcolorbox}[colback=blue!5!white, colframe=blue!50!black, title=Result 3: Subgroup mixture bound]
For any detector $\phi$ and any subgroup $G_k$, the subgroup-average false positive rate satisfies
\[
\bar{\alpha}_{G_k}(\phi) \;\geq\; \beta(\phi) - \,\mathrm{TV}(\bar{p}_{G_k},\, p_M).
\]
\end{tcolorbox}

If a subgroup's pooled writing distribution is close to AI output (small $\mathrm{TV}(\bar{p}_{G_k}, p_M)$), then \emph{any} detector with useful power must produce high false positive rates for that subgroup. This provides a theoretical explanation for the empirical findings of \citet{liang2023}, who documented elevated false positive rates for non-native English speakers, and \citet{hadra2026}, who found similar patterns for EFL learners.

The key mathematical step is the observation that the expected detector response under the subgroup mixture \emph{equals} the subgroup-average FPR. Aggregating across subgroups yields an institution-wide bound
\[
\text{Institution-wide average FPR} \;\geq\; \sum_{k=1}^K \pi_k \cdot \big(\beta(\phi) - \,\mathrm{TV}(\bar{p}_{G_k}, p_M)\big),
\]
where $\pi_k$ is the subgroup fraction of the student subgroup $G_k$.

\section{Discussion}
%=============================================================================
\subsection{Two Distinct Sources of Detection Difficulty}
\label{sec:twosources}
%=============================================================================

A key concept we describe in this paper is the separation of two logically independent mechanisms that make AI detection difficult, each with different policy implications.

Source 1 is AI quality convergence, where as LLMs improve, the statistical distance between AI output and human writing shrinks for \emph{all} human writers simultaneously. This is the mechanism analysed by \citet{sadasivan2023}, where intervention levers include degrading AI output quality, watermarking, or requiring AI providers to embed detectable signals.

Source 2 is population diversity discussed in this paper. Even with a \emph{fixed, non-improving} AI model, if the student population is diverse enough that some students' writing overlaps with AI output, any text-only, one-shot detector faces a fundamental size-power trade-off. Institutional intervention strategies include assessment redesign (open-ended tasks, personal reflection), process-based assessment (drafts, oral defence), building student writing profiles from multiple submissions.

These sources are complementary and independent. Crucially, Source~2 cannot be addressed by building a better detector: no text-only, one-shot detector can escape the size-power trade-off whenever distributional overlap exists. This is a consequence of the composite null, not a deficiency of any particular detector. Addressing Source~2 requires changing the \emph{assessment protocol}, i.e. collecting richer evidence than a single submitted document.

%=============================================================================
\subsection{Connecting Theory to Practice: Estimation and Auditing}
\label{sec:empirics}
%=============================================================================

Result~3 is useful insofar as the mixture TV distance $\mathrm{TV}(\bar{p}_{G_k}, p_M)$ can be related to observable quantities. We now discuss both the promise and the substantial difficulties of this step, with the aim of honest guidance for institutional practice.

A natural approach is to train a classifier to distinguish pooled subgroup documents from AI documents. If such a classifier achieves accuracy $\hat{a}$ under balanced classes, the Bayes-optimal accuracy satisfies $a^* = \frac{1}{2}(1 + \mathrm{TV})$, so
\[
\mathrm{TV}(\bar{p}_{G_k}, p_M) \;\geq\; 2\hat{a} - 1.
\]
This is a \emph{lower bound} on TV. However, for Result~3 to produce an informative lower bound on subgroup FPR, we need TV to be \emph{small} - that is, we need an \emph{upper bound} on TV, which goes in the opposite direction.

The implication is if a classifier struggles to distinguish a subgroup from AI output (low accuracy), this is \emph{suggestive} evidence that TV is small, but it is not proof - the classifier may simply be suboptimal. Specifically, failure of a particular classifier to exceed 50\% accuracy is not evidence that TV is small unless one can argue the classifier is near Bayes-optimal for the distributions in question. Conversely, if a classifier easily distinguishes the subgroup from AI, TV is large and the bound is vacuous for that subgroup, which is itself informative as it identifies subgroups where detection may be feasible.

%=============================================================================
\subsection{Task Dependence}
\label{sec:task}
%=============================================================================

All bounds depend on the assessment task $\tau$. The overlap parameters $\delta^*_\tau$, $\pi(\Theta^*_\tau)$, and $\mathrm{TV}(\bar{p}_{G_k,\tau}, p_{M,\tau})$ may vary substantially across tasks such as constrained-format tasks, for example ``write a 500-word summary of this article'', which reduce variability among student responses while simultaneously constraining AI output. Both effects increase the overlap between human and AI distributions, making detection structurally harder.

Open-ended or reflective tasks, for example ``discuss how your personal experience relates to this theory'', increase stylistic variability across students but also increase the effective TV distance between most students and the AI, making detection structurally easier for most (though not necessarily all) students.

The practical implication is that the \emph{same} detector may be structurally adequate for some task types and structurally inadequate for others, even within a single course. This reinforces an auditing protocol recommendation to evaluate detection reliability per-task group, not via a single institutional threshold.

%=============================================================================
\subsection{Limitations}
\label{sec:limitations}
%=============================================================================
The results we present are useful as an interpretive lens and help us understand and explain the outcomes observed from empirical research studies, as well as anecdotal and community observations reported throughout the higher education landscape. The assumptions invoked to simplify the AI detection problem, to allow probing statistical analysis and consideration are justified in the limit of some of the simplest ``black-box'' AI text detectors being implemented by higher education institutions globally. This being said, further improvements to the methodology and relaxation of certain assumptions may be possible in future iterations of analysis to extend the regime of validity to further detector scenarios. Below we briefly summarise some potential limitations or improvements which should be considered in future work.

\begin{enumerate}

\item \textbf{Estimation difficulty.} Result~3 provides a bound in terms of $\mathrm{TV}(\bar{p}_{G_k}, p_M)$. Estimating this quantity rigorously requires either near-Bayes-optimal classifiers or density estimation; classification accuracy yields only a lower bound on TV, which goes in the wrong direction for making the FPR bound informative (Section~\ref{sec:empirics}). We therefore recommend direct FPR auditing (Section~\ref{sec:audit}) rather than TV estimation. The difficulty of establishing ground truth for detection classification is well recognised in the educational measurement literature on test fraud~\citep{Meng2023}, where the absence of verified labels constrains the evaluation of both false positive and false negative rates—a challenge that the composite null framework formalises for AI text detection.

\item \textbf{Single-document model.} We treat each submission as a single draw from $p_\theta$, ignoring internal document structure. This is a reasonable assumption for a ``black-box'' one-shot, text-based AI detector commonly used in universities presently. Naturally other detector technologies are available, which can for example make decisions on more advanced document features, or longitudinal student data, though these are out of scope of this current work. 

\item \textbf{Fixed AI alternative.} Moreover, the boundary between ``AI-generated'' and ``AI-assisted human writing'' is not sharp as a student who uses an AI system for brainstorming or structural
suggestions before writing independently occupies a position on a continuum rather than falling cleanly under either hypothesis. Our binary formulation (human- written vs. AI-generated) is a simplification that mirrors the binary decisions some universities currently make, but the composite alternative problem suggests that the detection problem may be partially ill-posed at its boundaries.

% We model a single $p_M$. In practice, universities face a composite alternative (multiple AI models, prompting strategies, human-AI collaboration), which makes detection strictly harder than our bounds indicate.

\item \textbf{Scope of ``structural barrier'' claim.} What we have shown is that no text-only, one-shot detector can escape the size-power trade-off whenever distributional overlap exists between the student population and AI output. Alternative institutional designs that change the information structure - such as collecting multiple samples per student or building writing profiles, undertaking task redesign, or requiring process evidence - may reduce the effective overlap. Our results do not speak to these alternative protocols; they speak only to what is achievable from a single submitted document, in the one-shot, text-only environment overwhelmingly used currently in many higher education settings.
\end{enumerate}

\subsection{A Practical Path: Direct FPR Auditing}
\label{sec:audit}

In practice, the most operationally useful approach \emph{bypasses TV estimation entirely}. If a university applies a detector to a corpus of \emph{known human-written} documents from subgroup $G_k$ and observes a high false positive rate, this directly confirms that $\bar{\alpha}_{G_k}$ is large for that detector, without requiring estimation of the underlying TV distance. This is precisely the approach taken by \citet{liang2023}, \citet{weberwulff2023}, \citet{elkhatat2023}, and \citet{hadra2026}, and Result~3 provides the theoretical explanation: the observed high FPR is a \emph{necessary} consequence of the size-power trade-off when the subgroup mixture is close to $p_M$. A stratified auditing approach is consistent with established evaluation frameworks long established for automated scoring in educational measurement, which require evidence of subgroup-invariant performance before operational deployment~\citep{Williamson2012}.

\begin{tcolorbox}[colback=green!5!white, colframe=green!50!black, title=Proposed Institutional Auditing Protocol, fonttitle=\bfseries]
Before deploying any AI text detector, an institution should:

\begin{enumerate}
\item \textbf{Assemble stratified human-written corpora.} Collect confirmed human-written documents, stratified by observable subgroups that represent common student populations at the institution likely to exhibit different overlap with AI output: first language (L1 vs L2 speakers), diverse degree programmes, year of study, assessment task type, and so on.

\item \textbf{Run the detector on each stratum.} Apply the candidate detector to each subgroup corpus and record the false positive rate $\hat{\bar{\alpha}}_{G_k}$ for each subgroup $G_k$.

\item \textbf{Report stratified FPR alongside power.} For each subgroup, report the observed FPR \emph{together with} the detector's estimated power (from a parallel corpus of AI-generated text for the same task type). Result~3 predicts that subgroups with higher overlap will typically be expected to show higher FPR; if the observed trade-off is severe, the detector is structurally inadequate for that student subgroup/task combination.

\item \textbf{Evaluate per-task.} Repeat steps 1--3 for each assessment task type. Constrained tasks (e.g., ``summarise this article in 500 words'') are likely to produce higher overlap than open-ended or personal reflective tasks (Section~\ref{sec:task}).

\item \textbf{Set deployment thresholds per-stratum, not institution-wide.} If a detector passes audit for some subgroup/task combinations but not others, restrict deployment accordingly. Do not use a single institutional detection threshold which unfairly biases against sizeable student sub-populations.
\end{enumerate}

\medskip
This protocol tests the empirical prediction of Result~3 without requiring estimation of abstract TV distances (Section~\ref{sec:empirics}). It also generates the per-subgroup evidence base needed for equity-informed policy, as recommended by \citet{hadra2026} and \citet{liang2023}.
\end{tcolorbox}

%=============================================================================
\subsection{Implications for Policy and Practice}
\label{sec:policy}
%=============================================================================

Our results underscore several existing ideas emerging from research and in-practice empirical use of AI detection technology. We now highlight several concrete implications from our current work on how universities should approach AI text detection.

Despite the convenience and widespread adoption due to integrity concerns \citep{LuoJess2024}, AI detection scores should not serve as sole evidence. Results~1 and~2 establish that any text-only, one-shot detector with useful detection power will necessarily produce false accusations among students whose writing overlaps with AI output. Using detection scores as sole or primary evidence in misconduct proceedings therefore exposes institutions to a mathematically unavoidable risk of false findings. This provides a rigorous foundation for the policy recommendations already articulated on empirical grounds by \citet{weberwulff2023}, \citet{elkhatat2023}, and \citet{hadra2026}.

Population diversity is a constraint that is independent of, and cannot be addressed by improvements in, detector quality. For institutions with diverse student populations such as L2 English speakers, students from different disciplinary traditions, or students with varying levels of writing experience a structural barrier from population diversity may dominate the barrier from AI quality convergence. Investing in better detectors cannot remove this structural barrier within the text-only, one-shot setting. It is worth considering that investing in better assessment design may be beneficial in light of this.

Stratified auditing should be standard practice. Before deploying any detector, institutions should conduct the stratified FPR audit described in Section~\ref{sec:audit}. This is both feasible with existing resources and directly informative about whether the detector's error profile is acceptable for expected common subgroups and task types.

Assessment redesign is a durable solution, and may have added value and benefit of assuring validity of intended assessment goals \citep{Dawson2024, Corbin2025}. Given population diversity is a property of the composite null rather than of any detector, an effective institutional response is to change the information structure: process-based assessment, oral examination components, reflective tasks that draw on personal experience and developed understanding, and building longitudinal writing profiles that transform the composite null into a (partially) known null.

%=============================================================================
\section{Conclusion}
\label{sec:conclusion}
%=============================================================================

We have formalised the university ``black-box'', one-shot text-based AI detection problem as a composite hypothesis test and frame three quantitative results from fundamental mathematical statistics: average-case and worst-case trade-off bounds (Results~1-2), and a subgroup mixture bound (Result~3) that connects the theoretical framework to observable subgroup-level quantities. The central insight is that population diversity creates a structural barrier to text-only, one-shot AI detection that is logically independent of AI model quality: no text-only, one-shot detector can escape the size--power trade-off whenever distributional overlap exists between the student population and AI output.

This structural barrier provides a theoretical explanation for the patterns of unreliability and bias documented in the growing empirical literature on AI detection tools \citep{weberwulff2023,liang2023,elkhatat2023,hadra2026}. It also clarifies the scope of possible improvements, where better detector engineering cannot remove a barrier that arises from the diversity of the human population being tested. The durable institutional response is assessment redesign to change the information structure so that decisions need not rest on a single submitted document alone.

\bibliographystyle{apalike-ejor}
\bibliography{references}

%=============================================================================
% APPENDIX: Formal Statements and Proofs
%=============================================================================
\appendix
\section{Appendix: Formal Statements and Proofs}
\label{app:proofs}

For completeness, we state the theorems formally and provide proofs. All proofs use only the variational inequality~\eqref{eq:variational}.

\begin{theorem}[Average-case trade-off]
\label{thm:average}
Let $\{p_\theta\}_{\theta \in \Theta}$ be a family of distributions, $p_M$ a fixed distribution, and $\pi$ a probability measure on $\Theta$. For $\delta > 0$, let $\Theta^* = \{\theta : \mathrm{TV}(p_\theta, p_M) \leq \delta\}$ with $\pi(\Theta^*) > 0$. Then for any $\phi: \mathcal{X} \to [0,1]$:
\begin{equation}
\int_\Theta \alpha(\theta; \phi)\,d\pi(\theta) \;\geq\; \pi(\Theta^*) \cdot \big(\beta(\phi) - \delta\big).
\end{equation}
\end{theorem}

\begin{proof}
For $\theta \in \Theta^*$, by~\eqref{eq:variational}: $\alpha(\theta;\phi) = \mathbb{E}_{p_\theta}[\phi] \geq \mathbb{E}_{p_M}[\phi] - \,\mathrm{TV}(p_\theta, p_M) \geq \beta(\phi) - \delta$. Integrating over $\Theta$ and restricting to $\Theta^*$:
$\int_\Theta \alpha\,d\pi \geq \int_{\Theta^*}(\beta - \delta)\,d\pi = \pi(\Theta^*)(\beta - \delta)$.
\end{proof}

\begin{theorem}[Worst-case size--power bound]
\label{thm:worstcase}
Under the same setup, if $\sup_\theta \alpha(\theta;\phi) \leq \alpha_0$, then:
\begin{equation}
\beta(\phi) \;\leq\; \alpha_0 + \delta^*, \quad \delta^* = \inf_\theta \mathrm{TV}(p_\theta, p_M).
\end{equation}
\end{theorem}

\begin{proof}
Let $\theta^*$ achieve (or approach) the infimum. By~\eqref{eq:variational}: $\beta(\phi) - \alpha(\theta^*;\phi) \leq \delta^*$. Hence $\beta(\phi) \leq \alpha(\theta^*;\phi) + \delta^* \leq \alpha_0 + \delta^*$.
\end{proof}

\begin{theorem}[Subgroup mixture bound]
\label{thm:subgroup}
For subgroup $G_k$ with mixture $\bar{p}_{G_k} = \int_{G_k} p_\theta\,d\pi(\theta|G_k)$ and subgroup-average FPR $\bar{\alpha}_{G_k}(\phi) = \int_{G_k}\alpha(\theta;\phi)\,d\pi(\theta|G_k)$:
\begin{equation}
\bar{\alpha}_{G_k}(\phi) \;\geq\; \beta(\phi) - \,\mathrm{TV}(\bar{p}_{G_k}, p_M).
\end{equation}
\end{theorem}

\begin{proof}
By definition: $\mathbb{E}_{\bar{p}_{G_k}}[\phi] = \int_{G_k}\mathbb{E}_{p_\theta}[\phi]\,d\pi(\theta|G_k) = \bar{\alpha}_{G_k}(\phi)$. Applying~\eqref{eq:variational} to $(\bar{p}_{G_k}, p_M)$: $\beta(\phi) - \bar{\alpha}_{G_k}(\phi) \leq \,\mathrm{TV}(\bar{p}_{G_k}, p_M)$. Rearranging gives the result.
\end{proof}

\begin{remark}[Convexity]
By joint convexity of TV: $\mathrm{TV}(\bar{p}_{G_k}, p_M) \leq \int_{G_k}\mathrm{TV}(p_\theta, p_M)\,d\pi(\theta|G_k)$. The mixture TV is a lower bound on the average per-student TV within the subgroup. A small mixture TV does not imply all students are individually close to the AI.
\end{remark}

\end{document}